\documentclass[10pt]{article}
\usepackage{amsfonts}
\usepackage{amsfonts}
\usepackage{amsfonts}
\usepackage{amsfonts}
\usepackage{amsfonts}
\usepackage{amsfonts}
\usepackage{amsfonts}
\usepackage{amsfonts}
\usepackage{amsfonts}
\usepackage{amsfonts}
\usepackage{amsfonts}
\usepackage{amsfonts}
\usepackage{amsfonts}
\usepackage{amsfonts}
\usepackage{amsfonts}
\usepackage{amsfonts}
\usepackage{amsfonts}
\usepackage{amsfonts}
\usepackage{amsfonts}
\usepackage{amsfonts}
\usepackage{amsfonts}
\usepackage{amsfonts}
\usepackage{amsfonts}
\usepackage{amsfonts}
\usepackage{amsfonts}
\usepackage{amsfonts}
\usepackage{amsfonts}
\usepackage{amsfonts}
\usepackage{amsfonts}
\usepackage{amsfonts}
\usepackage{amsfonts}
\usepackage{amsfonts}
\usepackage{amsfonts}
\usepackage{amsfonts}
\usepackage{amsfonts}
\usepackage{amsfonts}
\usepackage{amsfonts}
\usepackage{amsfonts}
\usepackage{amsfonts}
\usepackage{amsfonts}
\usepackage{amsfonts}
\usepackage{amsfonts}
\usepackage{amsfonts}
\usepackage{amsfonts}
\usepackage{amsfonts}
\usepackage{amsfonts}
\usepackage{amsfonts}
\usepackage{amsfonts}
\usepackage{amsfonts}
\usepackage{amsfonts}
\usepackage{amsfonts}
\usepackage{amsfonts}
\usepackage{amsfonts}
\usepackage{amsfonts}
\usepackage{amsfonts}
\usepackage{amsfonts}
\usepackage{amsfonts}
\usepackage{amsfonts}
\usepackage{amsfonts}
\usepackage{amsfonts}
\usepackage{amsfonts}
\usepackage{amsfonts}
\usepackage{amsfonts}
\usepackage{amsfonts}
\usepackage{amsfonts}
\usepackage{amsfonts}
\usepackage{amsfonts}
\usepackage{amsfonts}
\usepackage{amsfonts}
\usepackage{amsfonts}
\usepackage{amsfonts}
\usepackage{amsfonts}
\usepackage{amsfonts}
\usepackage{amsfonts}
\usepackage{amsfonts}
\usepackage{amsfonts}
\usepackage{amsfonts}
\usepackage{amsfonts}
\usepackage{amsfonts}
\usepackage{amsfonts}
\usepackage{amsfonts}
\usepackage{amsfonts}
\usepackage{amsfonts}
\usepackage{amsfonts}
\usepackage{amsfonts}
\usepackage{amsfonts}
\usepackage{amsfonts}
\usepackage{amsfonts}
\usepackage{amsfonts}
\usepackage{amsfonts}
\usepackage{amsfonts}
\usepackage{amsfonts}
\usepackage{amsfonts}
\usepackage{amsfonts}
\usepackage{amsfonts}
\usepackage{amsfonts}
\usepackage{amsfonts}
\usepackage{amsfonts}
\usepackage{amsfonts}
\usepackage{amsfonts}
\usepackage{amsfonts}
\usepackage{amsfonts}
\usepackage{amsfonts}
\usepackage{amsfonts}
\usepackage{amsfonts}
\usepackage{amsfonts}
\usepackage{amsfonts}
\usepackage{amsfonts}
\usepackage{amsfonts}
\usepackage{amsfonts}
\usepackage{amsfonts}
\usepackage{amsfonts}
\usepackage{amsfonts}
\usepackage{amsfonts}
\usepackage{amsfonts}
\usepackage{amsfonts}
\usepackage{amsfonts}
\usepackage{amsfonts}
\usepackage{amsfonts}
\usepackage{amsfonts}
\usepackage{amsfonts}
\usepackage{amsfonts}
\usepackage{amsfonts}
\usepackage{amsfonts}
\usepackage{amsfonts}
\usepackage{amsfonts}
\usepackage{amsfonts}
\usepackage{amsfonts}
\usepackage{amsfonts}
\usepackage{amsfonts}
\usepackage{amsfonts}
\usepackage{amsfonts}
\usepackage{amsfonts}
\usepackage{amsfonts}
\usepackage{amsfonts}
\usepackage{amsfonts}
\usepackage{amsfonts}
\usepackage{amsfonts}
\usepackage{amsfonts}
\usepackage{amsfonts}
\usepackage{amsfonts}
\usepackage{amsfonts}
\usepackage{amsfonts}
\usepackage{amsfonts}
\usepackage{amsfonts}
\usepackage{amsfonts}
\usepackage{amsfonts}
\usepackage{amsfonts}
\usepackage{amsfonts}
\usepackage{amsfonts}
\usepackage{amsfonts}
\usepackage{amsfonts}
\usepackage{amsfonts}
\usepackage{amsfonts}
\usepackage{amsfonts}
\usepackage{amsfonts}
\usepackage{amsfonts}
\usepackage{amsfonts}
\usepackage{amsfonts}
\usepackage{amsfonts}
\usepackage{amsfonts}
\usepackage{amsfonts}
\usepackage{amsfonts}
\usepackage{amsfonts}
\usepackage{amsfonts}
\usepackage{amsfonts}
\usepackage{amsfonts}
\usepackage{amsfonts}
\usepackage{amsfonts}
\usepackage{amsfonts}
\usepackage{amsfonts}
\usepackage{amsfonts}
\usepackage{amsfonts}
\usepackage{amsfonts}
\usepackage{amsfonts}
\usepackage{amsfonts}
\usepackage{amsfonts}
\usepackage{amsfonts}
\usepackage{amsfonts}
\usepackage{amsfonts}
\usepackage{amsfonts}
\usepackage{amsfonts}
\usepackage{amsfonts}
\usepackage{amsfonts}
\usepackage{amsfonts}
\usepackage{amsfonts}
\usepackage{amsfonts}
\usepackage{amsfonts}
\usepackage{amsfonts}
\usepackage{amsfonts}
\usepackage{amsfonts}
\usepackage{amsfonts}
\usepackage{amsfonts}
\usepackage{amsfonts}
\usepackage{amsfonts}
\usepackage{amsfonts}
\usepackage{amsfonts}
\usepackage{amsmath}
\usepackage{amssymb, amscd, amsthm}
\usepackage[all]{xy}
\usepackage[dvips]{graphicx}
\usepackage{verbatim}
\usepackage{cite}
\usepackage{enumerate}
\usepackage{float}
\usepackage{diagbox}
\numberwithin{equation}{section}

\newtheorem{lemma}{\textbf{Lemma}}[section]
\newtheorem{theorem}{\textbf{Theorem}}
\newtheorem{remark}{\textbf{Remark}}[section]
\newtheorem{corollary}{\textbf{Corollary}}[section]

\newtheorem{example}{\textbf{Example}}[section]

\usepackage{longtable}

\usepackage{multirow}

\begin{document}
\baselineskip 17pt\title{\Large\bf MacWilliams Type Identities for Linear Block Codes on Certain Pomsets}

\author{\large  Wen Ma \quad\quad Jinquan Luo\footnote{The authors are with School of Mathematics
and Statistics \& Hubei Key Laboratory of Mathematical Sciences, Central China Normal University, Wuhan China.\newline
 E-mails: mawen95@126.com(W.Ma),  luojinquan@mail.ccnu.edu.cn(J.Luo)}}
\date{}
\maketitle

{\bf Abstract}: Pomset block metric is a generalization of pomset metric. In this paper, we define weight enumerator of linear block codes in pomset metric over $\mathbb{Z}_m$ and establish MacWilliams type identities for linear block codes with respect to certain pomsets. The relation between weight enumerators of two linear pomset block codes and their direct sum is also investigated.


\section{Introduction}

\quad\;The study of codes endowed with a metric other than the Hamming metric gained momentum since 1990's with the introduction of poset metric by Brialdi et al. (see [\ref{POSET}]). Feng et al. in [\ref{FENG}] introduced block metric and studied MDS block codes. Poset block metric was introduced by Alves et al. in [\ref{Alves}] unifying poset metric and block metric.

More recently in [\ref{POMSET}], Selvaraj and Sudha generalized the poset metric structure to a pomset metric structure. They introduced pomset metric and initialized the study of codes equipped with pomset metric. The concept of order ideals of a pomset is proposed and pomset metric is defined. We also refer the reader to [\ref{support}, \ref{PANEK}] for two general metrics in poset space.

Weight distribution of a code describes the number of codewords of each possible weight.  The MacWilliams identity for linear codes over finite fields is one of the most important identities in coding theory (see [\ref{ERROR}]). It expresses a connection between the weight enumerator of a linear codes and its dual code.

Kim and Oh classified all poset structures that admit the MacWilliams identity, and derived the MacWilliams identities for poset weight enumerators corresponding to such posets (see [\ref{POSETMAC}]). It is proved that being hierarchical is a necessary and sufficient condition for a poset to admit the MacWilliams identity. They also derived an explicit relation between the $P$-weight distribution of a hierarchical poset code and the $\bar{P}$-weight distribution of the dual code. Extending their observations, Pinheiro and Firer proved that a poset-block space admits a MacWilliams-type identity if and only if the poset is hierarchical, and at any level of the poset, all the blocks have the same dimension (see [\ref{posetmac}]). They explicitly stated the relation between the weight enumerators of a code and its dual when the poset-block admits the MacWilliams-type identity.

For codes over $\mathbb{Z}_m$, pomset metric is a generalization of Lee metric when pomset is taken to be an antichain; in some sense, it is a generalization to poset  metric as well. Sudha and Selvaraj defined  pomset weight enumerator of a code $\mathcal {C}$ and established MacWilliams type identities for linear codes with respect to certain pomsets (see [\ref{POMSETMAC}]). The identities for a particular type of linear codes are established by considering direct and ordinal sum of pomsets on them. For pomset block codes, it is natural to attempt to obtain some kind of MacWilliams identities for block codes in certain pomset metric. In this paper, we generalize the MacWilliams type identity given in [\ref{POMSETMAC}] for pomset spaces to pomset-block spaces.

The paper is organized as follows, Section 2 contains basic notions of pomset block metric over $\mathbb{Z}_m^n$ and defines weight enumerator for a linear pomset block code. In Section 3, we consider the relationship between the weight distribution of a pomset block code and its dual when the pomset is a chain pomset. In particular, an explicit relation is derived between the $\mathbb{P}$-weight distribution of a pomset block code and the $\widetilde{\mathbb{P}}$-weight distribution of the dual code when $\mathbb{Z}_m$ is a field and all blocks have dimension 2. We also give same examples to illustrate our conclusion. In Section 4, we give MacWilliams type identities on direct and ordinal sums of general pomsets. Lastly, we summarize our results and raise a question for further research.

\section{Preliminaries}

\quad\; In this section, we introduce some basic notations and useful results of a pomset block metric.

A collection of elements which may contain duplicates is called a \textbf{multiset} (in short, \textbf{mset}). Girish and John defined a multiset relation and explored some of basic properties (see [\ref{GIRISH}] and [\ref{girish}]).

Let $X$ be a set formally. A mset $M$ drawn from the set $X$ is represented by a function count $C_{M}:X\rightarrow\mathbb{N}$ where $\mathbb{N}$ represents the set of non-negative integers. For each $a\in X$, $C_{M}(a)$ indicates the number of occurrences of the element $a$ in $M$.

An element $a\in X$ appearing $p$ times in $M$ is denoted by $p/a\in M$ and thus $C_M(a)=p$. If we consider $k/a\in M$, the value of $k$ satisfies $k\leq p$. The mset drawn from the set $X=\{a_1,a_2,\ldots,a_n\}$ is represented as $M=\{p_1/a_1,p_2/a_2,\ldots,p_n/a_n\}$. The \textbf{cardinality} of an mset $M$ drawn from $X$ is $|M|=\sum_{a\in X}C_M(a)$. The \textbf{root set} of $M$ denoted by $M^{*}$ is defined as $M^{*}=\{a\in X:C_M(a)>0\}$.

Let $M_1$ and $M_2$ be two msets drawn from a set $X$. We call $M_1$ a \textbf{submset} of $M_2$ ($M_1\subseteq M_2$) if $C_{M_1}(a)\leq C_{M_2}(a)$ for all $a\in X$. The \textbf{union} of $M_1$ and $M_2$ is an mset denoted by $M=M_1\cup M_2$ such that for all $a\in X$, $C_M(a)=\text{max}\{C_{M_1}(a),C_{M_2}(a)\}$.

Let $M_1$ and $M_2$ be two msets drawn from $X$, the \textbf{Cartesian product} of $M_1$ and $M_2$ is also an mset defined as
$$M_1\times M_2=\{pq/(p/a,q/b):p/a\in M_1,q/b\in M_2\}.$$
Denote by $C_1(a,b)$ the count of the first coordinate in the ordered pair $(a,b)$ and by $C_2(a,b)$ the count of the second coordinate in the ordered pair $(a,b)$.

A submset $R$ of $M\times M$ is said to be an \textbf{mset relation} on $M$ if every member $(p/a,q/b)$ of $R$ has count $C_1(a,b)\cdot C_2(a,b)$. An mset relation $R$ on an mset $M$ is said to be \textbf{reflexive} if $m/a\ R\ m/a$ for all $m/a\in M$; \textbf{antisymmetric} if $m/a\ R\ n/b$ and $n/b\ R\ m/a$ imply $m=n$ and $a=b$; \textbf{transitive} if $m/a\ R\ n/b$ and $n/b\ R\ k/c$ imply $m/a\ R\ k/c$. An mset relation $R$ is called a \textbf{partially ordered mset relation (or order relation)} if it is reflexive, antisymmetric and transitive. The pair $(M,R)$ is known as a \textbf{partially ordered multiset (pomset)} denoted by $\mathbb{P}$.

Let $\mathbb{P}=(M,R)$ and $m/a\in M$. Then $m/a$ is a \textbf{maximal element} of $\mathbb{P}$ if there exists no $n/b\in M\ (b\neq a)$ such that $m/a\ R\ n/b$; $m/a$ is a \textbf{minimal element} if there exists no $n/b\in M\ (b\neq a)$ such that $n/b\ R\ m/a$. $\mathbb{P}$ is called a \textbf{chain} if every distinct pair of points from $M$ is comparable in $\mathbb{P}$. $\mathbb{P}$ is called an \textbf{anti-chain} if every distinct pair of points from $M$ is incomparable in $\mathbb{P}$.

A submset $I$ of $M$ is called an \textbf{order ideal} (or simply an \textbf{ideal}) of $\mathbb{P}$ if $k/a\in I$ and $q/b\ R\ k/a$ $(b\neq a)$ imply $q/b\in I$.  An \textbf{ideal generated by an element $k/a\in M$} is defined as
$$\langle k/a\rangle=\{k/a\}\cup\{q/b\in M: q/b\ R\ k/a\  \text{and}\ b\neq a\}. $$

An \textbf{ideal generated by a submset} $S$ of $M$ is defined by $\langle S\rangle=\bigcup\limits_{k/a\in S}\langle k/a\rangle$.

For a given pomset $\mathbb{P}=(M,R)$, the \textbf{dual pomset} $\widetilde{\mathbb{P}}=(M,\widetilde{R})$ of $\mathbb{P}$ is given by:
\begin{center}
$\mathbb{P}$ and $\widetilde{\mathbb{P}}$ have the same underlying set $M$ and $p/a\ R\ q/b$ in $\mathbb{P}$ if and only if $q/b\ \widetilde{R}\ p/a$ in $\widetilde{\mathbb{P}}$.
\end{center}
Note that $\mathbb{P}$ is a chain pomset implies that $\widetilde{\mathbb{P}}$ is a chain pomset.

Consider $\mathbb{Z}_m=\{0,1,\ldots,m-1\}$, the ring of integers modulo $m$. We consider a pomset $\mathbb{P}$ defined on an mset $M=\left\{\left\lfloor\frac{m}{2}\right\rfloor/1,\left\lfloor\frac{m}{2}\right\rfloor/2, \ldots,\left\lfloor\frac{m}{2}\right\rfloor/s\right\}$.

Let $\pi: [s]\rightarrow \mathbb{N}$ be a map such that $n=\sum\limits_{i=1}\limits^{s}\pi(i)$. The map $\pi$ is said to be a \textbf{labeling} of the pomset $\mathbb{P}$, and the pair $(\mathbb{P}, \pi)$ is called a \textbf{pomset block structure} over $[s]$. Denote $\pi(i)$ by $k_i$ and take $V_i$ as free $\mathbb{Z}_m$-module $\mathbb{Z}_m^{k_i}$ for all $1\leq i\leq s$. Define $V$ as
$$V=V_1\oplus V_2\oplus\cdots\oplus V_s$$
which is isomorphic to $\mathbb{Z}_m^{n}$. Each $\boldsymbol{u}\in V$ can be written as $$\boldsymbol{u}=(\boldsymbol{u_1},\boldsymbol{u_2},\ldots,\boldsymbol{u_s})$$
 where $\boldsymbol{u_i}=(u_{i1},u_{i2},\ldots,u_{ik_i})\in \mathbb{Z}_m^{k_i}$, $1\leq i\leq s$. For $a\in\mathbb{Z}_m$, \textbf{Lee weight} $w_L(a)$ of $a$ is minimum of $a$ and $m-a$.
The \textbf{Lee block support} of $\boldsymbol{u}\in V$ is defined as
$$supp_{(L,\pi)}(\boldsymbol{u})=\left\{s_i/i:s_i=w_{(L,\pi)}(\boldsymbol{u_i}),s_i\neq 0\right\},$$
where
$$w_{(L,\pi)}(\boldsymbol{u_i})=\text{max}\left\{w_L(\boldsymbol{u_{i_t}}):1\leq t\leq\pi(i)\right\}.$$

\noindent The \textbf{$(\mathbb{P},\pi)$-weight} of $\boldsymbol{u}\in V$ is defined to be the cardinality of the ideal generated by $supp_{(L,\pi)}(\boldsymbol{u})$, that is
$$w_{(\mathbb{P},\pi)}(\boldsymbol{u})=\left|\langle supp_{(L,\pi)}(\boldsymbol{u})\rangle\right|.$$

\noindent The \textbf{pomset block distance} between two vectors $\boldsymbol{u},\boldsymbol{v}\in V$ is given by
$$d_{(\mathbb{P},\pi)}(\boldsymbol{u},\boldsymbol{v})= w_{(\mathbb{P},\pi)}(\boldsymbol{u}-\boldsymbol{v})$$
which induces a metric on $\mathbb{Z}_m^n$ known as \textbf{pomset block metric}. The pair $\left(V, d_{(\mathbb{P},\pi)}\right)$ is said to be a \textbf{pomset block space}. A subset $\mathcal {C}$ of $\left(V,d_{(\mathbb{P},\pi)}\right)$ with cardinality $K$ is called an $(n,K,d)$ \textbf{$(\mathbb{P},\pi)$-code}, where $V$ is equipped with the pomset block metric $d_{(\mathbb{P},\pi)}(.,.)$ and
$$d=d_{(\mathbb{P},\pi)}(\mathcal {C})=\text{min}\left\{d_{(\mathbb{P},\pi)}(\boldsymbol{u},\boldsymbol{v}):\boldsymbol{u}\neq \boldsymbol{v}\in \mathcal {C}\right\}$$
is the \textbf{$(\mathbb{P},\pi)$-minimum distance} of $\mathcal {C}$. If $\mathcal {C}$ is a submodule of $V$ with cardinality $m^k$,  we call $\mathcal {C}$ a \textbf{linear $(n,m^k,d)$ $(\mathbb{P},\pi)$-code}.
The \textbf{dual} of an $(n,K,d)$ $(\mathbb{P},\pi)$-code $\mathcal {C}$ is defined as
$$\mathcal {C}^{\bot}=\left\{\boldsymbol{v}\in V:\boldsymbol{c} ¡¢
\cdot\boldsymbol{v}=c_1v_1+\cdots+c_nv_n=0\ \text{for all}\ \boldsymbol{c}\in \mathcal {C}\right\}.$$

For a linear $(\mathbb{P},\pi)$-code $\mathcal {C}$, the \textbf{$(\mathbb{P},\pi)$-weight enumerator} for $\mathcal {C}$ is the polynomial
$$W_{(\mathcal {C},\pi)}\left(x,y;\mathbb{P}\right)=\sum\limits_{\boldsymbol{u}\in\mathcal {C}}x^{s\lfloor\frac{m}{2}\rfloor-w_{(\mathbb{P},\pi)}(\boldsymbol{u})} y^{w_{(\mathbb{P},\pi)}(\boldsymbol{u})}=\sum \limits_{i=0}^{s\lfloor\frac{m}{2}\rfloor}A_{i,(\mathbb{P},\pi)}(\mathcal {C})x^{s\lfloor\frac{m}{2}\rfloor-i}y^i,$$
where $A_{i,(\mathbb{P},\pi)}(\mathcal {C})=\left|\{\boldsymbol{u}\in\mathcal {C}:w_{(\mathbb{P},\pi)}(\boldsymbol{u})=i\}\right|$.

\section{MacWilliams type identity in $(\mathbb{P},\pi)$ spaces for chain pomset}

\quad\;Let $[a]=\{1,2,\ldots,a\}$ and $[a,b]$ be the set of all integers between $a$ and $b$.  In this section, we will derive the MacWilliams type identity for linear block codes in the chain pomset metric. Without loss of generality, we define the pomset $\mathbb{P}=(M,R)$ on the multiset $M=\{\lfloor\frac{m}{2}\rfloor/1,\ldots,\lfloor\frac{m}{2}\rfloor/s\}$ whose mset relation is given by
$$\left\lfloor\frac{m}{2}\right\rfloor/i\ R\ \left\lfloor\frac{m}{2}\right\rfloor/j\Leftrightarrow i\leq j.$$

Let $(\mathbb{P},\pi)$ be a pomset block structure on $V$ where $\pi$ is a labeling of $\mathbb{P}$ such that $\sum\limits_{i=1}^s\pi(i)=n$. Suppose that $\boldsymbol{u}=(\boldsymbol{u_1},\ldots,\boldsymbol{u_s})\in V$ such that $\boldsymbol{u_i}\neq\textbf{0}\in\mathbb{Z}_m^{\pi(i)}$ and $(\boldsymbol{u_{i+1}},\boldsymbol{u_{i+2}},\ldots,\boldsymbol{u_s})= \textbf{0}\in\mathbb{Z}_m^{\pi(i+1)+\cdots+\pi(s)}$.
By the definition of $\mathbb{P}$, we have that $w_{(\mathbb{P},\pi)}(\boldsymbol{u})= \left\lfloor\frac{m}{2}\right\rfloor(i-1)+w_{(L,\pi)}(\boldsymbol{u_i})$. Given a linear $(\mathbb{P},\pi)$-code $\mathcal {C}$, we define
$$\mathcal {C}_i=\left\{\boldsymbol{u}\in\mathcal {C}:(\boldsymbol{u_i},\boldsymbol{u_{i+1}},\cdots,\boldsymbol{u_s})= \textbf{0}\in\mathbb{Z}_m^{\pi(i)+\pi(i+1)+\cdots+\pi(s)}\right\}$$
and
$$\mathcal {C}_i^{'}=\left\{\boldsymbol{u}\in\mathcal {C}:\textbf{0}\neq \boldsymbol{u_i}\in\mathbb{Z}_m^{\pi(i)},\ (\boldsymbol{u_{i+1}},\cdots,\boldsymbol{u_s})= \textbf{0}\in\mathbb{Z}_m^{\pi(i+1)+\cdots\pi(s)}\right\}.$$

Let $\mathcal {R}$ be a finite ring. Recall that an \textbf{additive character} $\chi$ on $\mathcal {R}$ is just a group homomorphism from the additive group $\mathcal {R}$ into the multiplicative group $\mathbb{C}^{*}$. The set of all additive character of $\mathcal {R}$ forms a group $\hat{\mathcal {R}}$, called the \textbf{character group} whose group operation is the pointwise multiplication of characters. Moreover, $\hat{\mathcal {R}}$ is a right $\mathcal {R}$-module with the function $\mathcal {R}\times\hat{\mathcal {R}}\rightarrow\hat{\mathcal {R}}$ given by $(a,\chi)=\chi_a$ where $\chi_a\in\hat{\mathcal {R}}$ such that $\chi_a(b)=\chi(ab)$ for all $\chi\in\hat{\mathcal {R}}$ and $a,b\in \mathcal {R}$. A character $\chi$ of $\mathcal {R}$ is a \textbf{right generating character} if the mapping $\phi:\mathcal {R}\rightarrow\hat{\mathcal {R}}$ given by $\phi(r)=\chi_r$ is an isomorphism of $\mathcal {R}$-modules. See [\ref{CHARACTER}] for detailed discussion on additive characters.

\begin{lemma}([\ref{CHARACTER}])\label{character}
Let $\chi$ be a character of a finite ring $\mathcal {R}$. Then $\chi$ is a right generating character if and only if ker $\chi$ contains no non-zero right ideals.
\end{lemma}

The following lemmas are easy consequences of Lemma \ref{character}.

\begin{lemma}\label{ring}
Let $\chi$ be a nontrivial additive character on a finite communicative ring $\mathcal {R}$ and $\boldsymbol{a}$ be a fixed element of $\mathcal {R}$. Then
$$\sum\limits_{\boldsymbol{b}\in \mathcal {R}}\chi(\boldsymbol{a}\cdot\boldsymbol{b})=\left\{
                             \begin{array}{ll}
                               |\mathcal {R}|, & \text{if}\ \boldsymbol{a}=\boldsymbol{0};\\[3mm]
                               0,   & \text{if}\ \boldsymbol{a}\neq \boldsymbol{0}.
                             \end{array}
                           \right.$$
\end{lemma}

\begin{lemma}
Let $\chi$ be a generating character on a finite communicative ring $\mathcal {R}$. For any submodule $\mathcal {C}\subseteq\mathcal {R}^n$, we have
$$\sum\limits_{\boldsymbol{u}\in \mathcal {R}}\chi(\boldsymbol{u}\cdot\boldsymbol{v})=\left\{
                             \begin{array}{ll}
                               |\mathcal {C}|, & \text{if}\ \boldsymbol{v}\in\mathcal {C}^{\bot};\\[3mm]
                               0,  & \text{if}\ \boldsymbol{v}\notin\mathcal {C}^{\bot}.
                             \end{array}
                           \right.$$
\end{lemma}

Let $f$ be a complex-valued function defined on $\mathcal {R}^n$ and $\chi$ be a generating character of $\mathcal {R}$. The Fourier transform of $f$ is
$$\hat{f}(\boldsymbol{u})=\sum\limits_{\boldsymbol{v}\in\mathcal {R}^n}\chi(\boldsymbol{u}\cdot\boldsymbol{v})f(\boldsymbol{v}).$$

\begin{lemma}\label{dual}
Let $\mathcal {C}\subseteq\mathcal {R}^n$ be a submodule and $f$ be a function defined on $\mathcal {R}^n$. Then
$$\sum\limits_{\boldsymbol{v}\in\mathcal {C}^{\bot}}f(\boldsymbol{v})=\frac{1}{|\mathcal {C}|}\sum\limits_{\boldsymbol{u}\in\mathcal {C}}\hat{f}(\boldsymbol{u}).$$
\end{lemma}

\begin{theorem}\label{MAC}
Given a linear $(\mathbb{P},\pi)$-code $\mathcal {C}$ of length $n$ over $\mathbb{Z}_m$ on the chain pomset $\mathbb{P}=(M,R)$, we have the following
\begin{enumerate}[(1)]
\item if $m$ is odd then
  \begin{equation*}\label{equ1} W_{(\mathcal{C}^{\bot},\pi)}(x,y;\widetilde{\mathbb{P}})=x^{s\lfloor\frac{m}{2}\rfloor}+ \sum\limits_{i=1}^s \frac{2m^{\pi(i+1)+\cdots+\pi(s)}}{|\mathcal {C}|}\sum \limits_{j=1}^{\lfloor\frac{m}{2}\rfloor} \left(\frac{y}{x}\right)^{(s-i)\lfloor\frac{m}{2}\rfloor+j}\left[\beta_{ij} W_{(\mathcal {C}_i,\pi)}(x,x;\mathbb{P})+LW_{\mathcal {C}_i^{'},\pi}^{j}x^{s\lfloor\frac{m}{2}\rfloor}\right];
  \end{equation*}

\item if $m$ is even then
  \begin{equation*}
  \begin{aligned}
  & W_{(\mathcal{C}^{\bot},\pi)}(x,y;\widetilde{\mathbb{P}})=x^{s\frac{m}{2}}+ \sum\limits_{i=1}^s \frac{m^{\pi(i+1)+\cdots+\pi(s)}}{|\mathcal {C}|}\left(\frac{y}{x}\right)^{(s-i)\frac{m}{2}}&\\
  &\left[\left(2\sum\limits_{j=1}^{\frac{m}{2}-1}\beta_{ij}\left(\frac{y}{x}\right)^j+ \gamma_i\left(\frac{y}{x}\right)^{\frac{m}{2}}\right)W_{(\mathcal {C}_i,\pi)}(x,x;\mathbb{P})+\left(2\sum\limits_{j=1}^{\frac{m}{2}-1}\left(\frac{y}{x}\right)^j LW_{\mathcal {C}_i^{'},\pi}^j+\left(\frac{y}{x}\right)^{\frac{m}{2}}LW_{\mathcal {C}_i^{'},\pi}^{\frac{m}{2}}\right)x^{s\frac{m}{2}}\right],&\\
  \end{aligned}
  \end{equation*}
  where $\beta_{ij}=\frac{(2j+1)^{\pi(i)}-(2j-1)^{\pi(i)}}{2}$, $\gamma_i=m^{\pi(i)}-(m-1)^{\pi(i)}$ and
  $$\left\{\begin{array}{ll}
  LW_{\mathcal {C}_i^{'},\pi}^j=\sum\limits_{\boldsymbol{u}\in\mathcal {C}_i^{'}}\sum\limits_{a=1}^{\pi(i)}\cos\frac{2\pi u_{i_a}j}{m}\prod\limits_{b<a}\left(1+2\sum\limits_{t=1}^{j-1}\cos\frac{2\pi u_{i_b}t}{m}\right)\prod\limits_{b>a}\left(1+2\sum\limits_{t=1}^{j}\cos\frac{2\pi u_{i_b}t}{m}\right)&1\leq j\leq\lfloor\frac{m-1}{2}\rfloor;\\[5mm]
  LW_{\mathcal {C}_i^{'},\pi}^{\frac{m}{2}}=\sum\limits_{\boldsymbol{u}\in\mathcal {C}_i^{'}}\sum\limits_{a=1}^{\pi(i)}(-1)^{u_{i_a}} \prod\limits_{b<a}\left(1+2\sum\limits_{t=1}^{\frac{m}{2}-1}\cos\frac{2\pi u_{i_b}t}{m}\right)\prod\limits_{b>a}\left(1+(-1)^{u_{i_b}}+ 2\sum\limits_{t=1}^{\frac{m}{2}-1}\cos\frac{2\pi u_{i_b}t}{m}\right)& j=\frac{m}{2}\in\mathbb{Z}.
  \end{array}
  \right.$$

\end{enumerate}
\end{theorem}

\begin{proof}
Consider the function $f: V\rightarrow\mathcal {C}(x,y)$ defined by
$$f(\boldsymbol{u})= x^{s\lfloor\frac{m}{2}\rfloor-w_{(\widetilde{\mathbb{P}},\pi)}(\boldsymbol{u})} y^{w_{(\widetilde{\mathbb{P}},\pi)}(\boldsymbol{u})}.$$
Then it follows from Lemma \ref{dual} that
\begin{eqnarray}\label{e31}
W_{(\mathcal {C}^{\bot},\pi)}(x,y;\widetilde{\mathbb{P}})=\sum\limits_{\boldsymbol{u}\in\mathcal {C}^{\bot}}f(\boldsymbol{u})=\frac{1}{|\mathcal {C}|}\sum\limits_{\boldsymbol{u}\in\mathcal {C}}\hat{f}(\boldsymbol{u}).
\end{eqnarray}
We now analyze the value $\hat{f}(\boldsymbol{u})$ in detail. We have that
\begin{eqnarray*}
\hat{f}(\boldsymbol{u})&=&\sum\limits_{\boldsymbol{v}\in V}\chi(\boldsymbol{u}\cdot\boldsymbol{v}) x^{s\lfloor\frac{m}{2}\rfloor-w_{(\widetilde{\mathbb{P}},\pi)}(\boldsymbol{v})} y^{w_{(\widetilde{\mathbb{P}},\pi)}(\boldsymbol{v})}\\
&=&x^{s\lfloor\frac{m}{2}\rfloor}\left[1+\sum\limits_{\textbf{0}\neq \boldsymbol{v}\in\ V}\chi(\boldsymbol{u}\cdot\boldsymbol{v})\left(\frac{y}{x}\right)^{w_{(\widetilde{\mathbb{P}},\pi)(\boldsymbol{v})}}\right].
\end{eqnarray*}
Given $i\in[s]$, $j\in\left[\lfloor\frac{m}{2}\rfloor\right]$ and $a\in[\pi(i)]$, set:
$$\left\{
 \begin{array}{l}
 D_{i}=\{\boldsymbol{v}\in V\setminus\{\textbf{0}\}: \min\{k:\boldsymbol{v_k}\neq \textbf{0}\}=i\};\\[3mm]
 E_{ij}=\{\boldsymbol{v}\in D_i:w_{(L,\pi)}(\boldsymbol{v_i})=j\};\\[3mm]
 E_{ija}=\{\boldsymbol{v}\in E_{ij}: \min\{k:w_L(v_{i_k})=j\}=a\}.
 \end{array}
 \right.$$
With these definitions, we have
\begin{eqnarray*}
\hat{f}(\boldsymbol{u})&=& x^{s\lfloor\frac{m}{2}\rfloor}\left[1+\sum\limits_{i=1}^s\sum\limits_{\boldsymbol{v}\in D_i}\chi(\boldsymbol{u}\cdot\boldsymbol{v}) \left(\frac{y}{x}\right)^{(s-i)\lfloor\frac{m}{2}\rfloor+ w_{(L,\pi)}(\boldsymbol{v_i})}\right]\\
&=&x^{s\lfloor\frac{m}{2}\rfloor} \left[1+\sum\limits_{i=1}^s\sum\limits_{j=1}^{\lfloor\frac{m}{2}\rfloor} \sum\limits_{\boldsymbol{v}\in E_{ij}}\left(\frac{y}{x}\right)^{(s-i)\lfloor\frac{m}{2}\rfloor+j} \chi(\boldsymbol{u_i}\cdot\boldsymbol{v_i}+ \cdots+\boldsymbol{u_s}\cdot\boldsymbol{v_s})\right]\\
&=&x^{s\lfloor\frac{m}{2}\rfloor} \left[1+\sum\limits_{i=1}^s\sum\limits_{j=1}^{\lfloor\frac{m}{2}\rfloor} \left(\frac{y}{x}\right)^{(s-i)\lfloor\frac{m}{2}\rfloor+j}\sum\limits_{\boldsymbol{v}\in E_{ij}}\chi(\boldsymbol{u_i}\cdot\boldsymbol{v_i}+\cdots+ \boldsymbol{u_s}\cdot\boldsymbol{v_s})\right]\\
&=&x^{s\lfloor\frac{m}{2}\rfloor} \left[1+\sum\limits_{i=1}^s\sum\limits_{j=1}^{\lfloor\frac{m}{2}\rfloor} \left(\frac{y}{x}\right)^{(s-i)\lfloor\frac{m}{2}\rfloor+j} \sum\limits_{\textbf{0}\neq \boldsymbol{v_i}\in V_i\atop w_{(L,\pi)}(\boldsymbol{v_i})=j}\chi(\boldsymbol{u_i}\cdot\boldsymbol{v_i}) \sum\limits_{\boldsymbol{v'}\in\mathbb{Z}_m^{\pi(i+1)+\cdots+\pi(s)}} \chi(\boldsymbol{u'}\cdot\boldsymbol{v'})\right],
\end{eqnarray*}
where $\boldsymbol{u'}=(\boldsymbol{u_{i+1}},\ldots,\boldsymbol{u_s}) \in\mathbb{Z}_m^{\pi(i+1)+\cdots+\pi(s)}$. It follows from Lemma \ref{ring} that
\begin{eqnarray}\label{e32}
\sum\limits_{\boldsymbol{u}\in\mathcal {C}}\hat{f}(\boldsymbol{u})=|\mathcal {C}|x^{s\lfloor\frac{m}{2}\rfloor}+\sum\limits_{i=1}^sx^{i\lfloor\frac{m}{2}\rfloor} y^{(s-i)\lfloor\frac{m}{2}\rfloor}m^{\pi(i+1)+\cdots+\pi(s)}\sum\limits_{j=1}^{\lfloor\frac{m}{2}\rfloor} \left(\frac{y}{x}\right)^{j}\sum\limits_{\boldsymbol{u}\in \mathcal {C}_i\cup\mathcal {C}_i^{'}}\sum\limits_{\textbf{0}\neq \boldsymbol{v_i}\in V_i\atop w_{(L,\pi)}(\boldsymbol{v_i})=j}\chi(\boldsymbol{u_i}\cdot\boldsymbol{v_i}).
\end{eqnarray}
Set
$$Z_i=\sum\limits_{j=1}^{\lfloor\frac{m}{2}\rfloor}\left(\frac{y}{x}\right)^{j} \sum\limits_{\boldsymbol{u}\in \mathcal {C}_i\cup\mathcal {C}_i^{'}} \sum\limits_{\textbf{0}\neq \boldsymbol{v_i}\in V_i\atop w_{(L,\pi)}(\boldsymbol{v_i})=j}\chi(\boldsymbol{u_i}\cdot\boldsymbol{v_i}).$$
If $m$ is odd then we observe that
\begin{eqnarray*}
Z_i&=&\sum\limits_{j=1}^{\lfloor\frac{m}{2}\rfloor}\left(\frac{y}{x}\right)^{j} \sum\limits_{\boldsymbol{u}\in \mathcal {C}_i\cup\mathcal {C}_i^{'}}\sum\limits_{a=1}^{\pi(i)}\sum\limits_{\boldsymbol{v}\in E_{ija}}\chi(\boldsymbol{u_i}\cdot\boldsymbol{v_i})\\
&=&\sum\limits_{j=1}^{\lfloor\frac{m}{2}\rfloor}\left(\frac{y}{x}\right)^{j} \sum\limits_{\boldsymbol{u}\in \mathcal {C}_i\cup\mathcal {C}_i^{'}}\sum\limits_{a=1}^{\pi(i)}\sum\limits_{\boldsymbol{v_i}\in V_i, w_L(v_{i_a})=j;\atop w_L(v_{i_b})<j\ \text{for}\ b<a; w_{L}(v_{i_b})\leq j\ \text{for}\ b>a}\chi(\boldsymbol{u_i}\cdot\boldsymbol{v_i})\\
&=&2\sum\limits_{j=1}^{\lfloor\frac{m}{2}\rfloor}\left(\frac{y}{x}\right)^{j} \sum\limits_{\boldsymbol{u}\in \mathcal {C}_i\cup\mathcal {C}_i^{'}}\sum\limits_{a=1}^{\pi(i)}\cos\frac{2\pi u_{i_a}j}{m}\prod\limits_{b<a}\left(1+2\sum\limits_{t=1}^{j-1}\cos\frac{2\pi u_{i_b}t}{m}\right)\prod\limits_{b>a}\left(1+2\sum\limits_{t=1}^{j}\cos\frac{2\pi u_{i_b}t}{m}\right)\\
&=&2\sum\limits_{j=1}^{\lfloor\frac{m}{2}\rfloor}\left(\frac{y}{x}\right)^{j}\Bigg[ |\mathcal {C}_i|\frac{(2j+1)^{\pi(i)}-(2j-1)^{\pi(i)}}{2}+\\
&&\sum\limits_{\boldsymbol{u}\in\mathcal {C}_i^{'}}\sum\limits_{a=1}^{\pi(i)}\cos\frac{2\pi u_{i_a}j}{m}\prod\limits_{b<a}\left(1+2\sum\limits_{t=1}^{j-1}\cos\frac{2\pi u_{i_b}t}{m}\right)\prod\limits_{b>a}\left(1+2\sum\limits_{t=1}^{j}\cos\frac{2\pi u_{i_b}t}{m}\right)\Bigg].
\end{eqnarray*}
The result then follows from (\ref{e31}) and (\ref{e32}). The case for $m$ even can be proved in a similar way.
\end{proof}

\begin{remark}
Note that when we consider the pomset block metric over $\mathbb{Z}_2^s$ and $\mathbb{Z}_3^s$, the pomset block metric will coincide with poset block metric. Theorem \ref{MAC} is consistent with the result in [\ref{posetmac}].
\end{remark}

\begin{example}
Let $M=\{2/1,2/2\}$ and $\mathbb{P}=(M,R)$ be a pomset whose order relation is chain relation. Let $\pi$ be a labeling of the pomset $\mathbb{P}$ such that $\pi(1)=2$ and $\pi(2)=1$. Consider the $(\mathbb{P},\pi)$-code $\mathcal {C}\subseteq\mathbb{Z}_4^3$ given by
$$\mathcal {C}=\{000, 112, 220, 332\}.$$
Then by Theorem \ref{MAC}, one has
\begin{eqnarray*}
W_{(\mathcal {C}^{\bot},\pi)}(x,y;\widetilde{\mathbb{P}})&=&x^4+ \sum\limits_{i=1}^2\frac{4^{\pi(i+1)+\cdots+\pi(2)}}{4}\left(\frac{y}{x}\right)^{4-2i}\\
&&\left[\left(2\beta_{i1}\frac{y}{x}+\gamma_i\left(\frac{y}{x}\right)^2\right)W_{(\mathcal {C}_i,\pi)}(x,x;\mathbb{P})+\left(2\left(\frac{y}{x}\right)LW_{\mathcal {C}_i^{'},\pi}^1+\left(\frac{y}{x}\right)^2LW_{\mathcal {C}_i^{'},\pi}^{2}\right)x^4\right].
\end{eqnarray*}
Note that $\mathcal {C}_1=\{000\}$, $\mathcal {C}_2=\{000,220\}$, $\mathcal {C}_1^{'}=\{220\}$ and $\mathcal {C}_2^{'}=\{112,332\}$. Hence
\begin{eqnarray*}
W_{(\mathcal {C}^{\bot},\pi)}(x,y;\widetilde{\mathbb{P}})&=&x^4+x^2y^2+8xy^3+6y^4.
\end{eqnarray*}
On the other hand, the dual of $\mathcal {C}$ is
$$\mathcal {C}^{\bot}=\{000, 111, 222, 333, 130, 220, 310, 021, 002, 023, 201, 203, 113, 132, 312, 331\}.$$
The $(\widetilde{\mathbb{P}},\pi)$-weight enumerator for $\mathcal {C}^{\bot}$ is then
$$W_{(\mathcal {C}^{\bot},\pi)}(x,y;\widetilde{\mathbb{P}})=x^4+x^2y^2+8xy^3+6y^4,$$
which coincides with Theorem \ref{MAC}.
$\hfill\square$
\end{example}

\begin{example}
Let $M=\{2/1,2/2\}$ and $\mathbb{P}=(M,R)$ be a pomset whose order relation is chain relation. Let $\pi$ be a labeling of the pomset $\mathbb{P}$ such that $\pi(1)=1$ and $\pi(2)=2$. Consider the $(\mathbb{P},\pi)$-code $\mathcal {C}\subseteq\mathbb{Z}_5^3$ given by
$$\mathcal {C}=\{000,132,214,341,423\}.$$
Then we have $\mathcal {C}_1=\{000\}=\mathcal {C}_2$, $\mathcal {C}_1^{'}=\emptyset$ and $\mathcal {C}_2^{'}=\{132, 214, 341, 423\}$. It follows from Theorem \ref{MAC} that
\begin{eqnarray*}
W_{(\mathcal {C}^{\bot},\pi)}(x,y;\widetilde{\mathbb{P}})&=& x^4+\sum\limits_{i=1}^2\frac{2\times 5^{\pi(i+1)+\cdots+\pi(2)}}{5}\sum\limits_{j=1}^{2}\left(\frac{y}{x}\right)^{4-2i+j} \left[\beta_{ij}W_{(\mathcal {C}_i,\pi)}(x,x;\mathbb{P})+LW_{\mathcal {C}_i^{'},\pi}^jx^4\right]\\
&=&x^4+2x^3y+2x^2y^2+10xy^3+10y^4.
\end{eqnarray*}
On the other hand, the dual of $\mathcal {C}$ is
$$\mathcal {C}^{\bot}=\{000,102,204,301,403,011,022,033,044,113,221,334,$$
$$442,124,243,312,431,130,210, 340,420,141,232,323,414\}.$$
The $(\widetilde{\mathbb{P}},\pi)$-weight enumerator for $\mathcal {C}^{\bot}$ is then
$$W_{(\mathcal {C}^{\bot},\pi)}(x,y;\widetilde{\mathbb{P}})=x^4+2x^3y+2x^2y^2+10xy^3+10y^4,$$
which also coincides with Theorem \ref{MAC}.
$\hfill\square$
\end{example}

\begin{remark}
Let $\mathbb{P}=(M,R)$ be a pomset on $M=\left\{\lfloor\frac{m}{2}\rfloor/1,\lfloor\frac{m}{2}\rfloor/2,\ldots, \lfloor\frac{m}{2}\rfloor/s\right\}$ with chain relation and let $\pi$ be a labeling of the pomset $\mathbb{P}$ with $\pi(i)=1$ for $i\in[s]$. Let $\mathcal {C}$ be a linear $(\mathbb{P},\pi)$-code. Then the $(\widetilde{\mathbb{P}},\pi)$-weight enumerator for $\mathcal {C}^{\bot}$ is given by the followings.
\begin{itemize}
\item when $m$ is odd,
$$W_{(\mathcal {C}^{\bot},\pi)}(x,y;\widetilde{\mathbb{P}})=x^{s\lfloor\frac{m}{2}\rfloor}+ \sum\limits_{i=1}^{s}\frac{2m^{s-i}}{|\mathcal {C}|}\sum\limits_{j=1}^{\lfloor\frac{m}{2}\rfloor}\left(\frac{y}{x}\right)^ {(s-i)\lfloor\frac{m}{2}\rfloor+j}\left(W_{(\mathcal {C}_i,\pi)}(x,x;\mathbb{P})+\sum\limits_{l=1}^{\lfloor\frac{m}{2}\rfloor}\cos\frac{2\pi lj}{m}W_{\mathcal {C}_{il}^{'}}(x,x;\mathbb{P})\right).$$

\item when $m$ is even,
$$W_{(\mathcal {C}^{\bot},\pi)}(x,y;\widetilde{\mathbb{P}})=x^{s\lfloor\frac{m}{2}\rfloor}+ \sum\limits_{i=1}^s\frac{m^{s-i}}{|\mathcal {C}|}\left(\frac{y}{x}\right)^{(s-i)\frac{m}{2}}$$ $$\left[\left(2\sum\limits_{j=1}^{\frac{m}{2}-1}\left(\frac{y}{x}\right)^j+ \left(\frac{y}{x}\right)^{\frac{m}{2}}\right)W_{(\mathcal {C}_i,\pi)}(x,x;\mathbb{P})+ \left(2\sum\limits_{j=1}^{\frac{m}{2}-1}\left(\frac{y}{x}\right)^j \sum\limits_{u\in\mathcal {C}_i^{'}}\cos(\frac{2\pi u_{i}j}{m})+\left(\frac{y}{x}\right)^{\frac{m}{2}}\sum\limits_{u\in\mathcal {C}_i^{'}}(-1)^{u_i}\right)x^{\frac{sm}{2}}\right].$$
\end{itemize}
The result is the same as the MacWilliams type identity for the case of pomset metric (see [\ref{POMSETMAC}]).
\end{remark}

Before we proceed to prove the next corollary, we assume the following notations:
\begin{enumerate}[(1)]
\item Given $j\in\left[\lfloor\frac{m}{2}\rfloor\right]$ even. Consider the $i$-th block $(u_{i_1},u_{i_2})$ of $\boldsymbol{u}\in\mathcal {C}$.
\begin{itemize}
\item  $\mathcal {C}_{ij}^1=\{\boldsymbol{u}\in\mathcal {C}_i^{'}: (u_{i_1},u_{i_1})=(1,\pm t)\}$, where $t$ satisfies one of the following conditions:
\begin{enumerate}[a)]
\item $\frac{lm+1}{j}-1\leq t\leq\frac{lm-1}{j}+1\ (1\leq l\leq \frac{j}{2}-1),\ t\mid j\ \text{or}\ t\nmid j\ \text{but}\ \exists\ a\in[-j,j]\ s.t.\ at+j\equiv 0\ (\text{mod}\ m);$
\item $\frac{jm+2}{2j}-1\leq t\leq\frac{m-1}{2}, \ t\mid j\ \text{or}\ t\nmid j\ \text{but}\ \exists\ a\in[-j,j]\ s.t.\ at+j\equiv 0\ (\text{mod}\ m).$
\end{enumerate}

\item  $\mathcal {C}_{ij}^2=\{\boldsymbol{u}\in\mathcal {C}_i^{'}: (u_{i_1},u_{i_1})=(1,\pm t)\}$, where $t$ satisfies one of the following conditions:
\begin{enumerate}[a)]
\item $1\leq t\leq\frac{m+1}{j}-1$, $t\nmid j$;
\item $\frac{lm-1}{j}+1<t<\frac{(l+1)m+1}{j}-1$ $(1\leq l\leq\frac{j}{2}-1)$, $t\nmid j$ and there does not exist $a\in[-j,j]\ s.t.\ at+j\equiv 0$ (mod $m$).
\end{enumerate}

\item $\mathcal {C}_{ij}^3=\{\boldsymbol{u}\in\mathcal {C}_i^{'}: (u_{i_1},u_{i_2})=(1,\pm t),\ \boldsymbol{u}\notin \mathcal {C}_{ij}^{1}\cup\mathcal {C}_{ij}^{2}\}\cup\{\boldsymbol{u}\in\mathcal {C}_{i}^{'}: (u_{i_1},u_{i_2})=(0,1)\}$.
\end{itemize}

\item Given $j\in\left[3,\lfloor\frac{m}{2}\rfloor\right]$ odd. Consider the $i$-th block $(u_{i_1},u_{i_2})$ of $\boldsymbol{u}\in\mathcal {C}$.
\begin{itemize}
\item $\mathcal {C}_{ij}^1=\{\boldsymbol{u}\in\mathcal {C}_i^{'}:(u_{i_1},u_{i_2})=(1,\pm t)\}$ where $t$ satisfies
$$\frac{lm+1}{j}-1\leq t\leq\frac{lm-1}{j}+1\ (1\leq l\leq \frac{j-1}{2}),\ t\mid j\ \text{or}\ t\nmid j\ \text{but}\ \exists\ a\in[-j,j]\ s.t.\ at+j\equiv 0\ (\text{mod}\ m);$$
\item $\mathcal {C}_{ij}^2=\{\boldsymbol{u}\in\mathcal {C}_i^{'}:(u_{i_1},u_{i_2})=(1,\pm t)\}$ where $t$ satisfies one of the following conditions:
\begin{enumerate}[a)]
\item $1\leq t<\frac{m+1}{j}-1$, $t\nmid j$;
\item $\frac{lm-1}{j}+1\leq t\leq \frac{(l+1)m+1}{j}-1$ ($1\leq l\leq\frac{j-3}{2}$), $t\nmid j$ and there does not exist $a\in[-j,j]\ s.t. \ at+j\equiv 0$ (mod $m$).
\item $\frac{(j-1)m-2}{2j}+1<t\leq\frac{m-1}{2}$, $t\nmid j$ and there exists no $a\in[-j,j]\ s.t.\ at+j\equiv 0$ (mod $m$).
\end{enumerate}
\item $\mathcal {C}_{ij}^3=\{\boldsymbol{u}\in\mathcal {C}_i^{'}:(u_{i_1},u_{i_2})=(1,\pm t),\ \boldsymbol{u}\notin \mathcal {C}_{ij}^1\cup\mathcal {C}_{ij}^2\}\cup\{\boldsymbol{u}\in\mathcal {C}_i^{'}: (u_{i_1},u_{i_2})=(0,1)\}$.
\end{itemize}

\item For $j=1$, we put $\mathcal {C}_{i1}^1=\emptyset$,
$$\mathcal {C}_{i1}^2=\left\{\boldsymbol{u}\in\mathcal {C}_i^{'}:(u_{i_1},u_{i_2})=(1,t), t\in\mathbb{Z}_m\setminus \{\pm1,0\}\right\},$$
and
$$\mathcal {C}_{i1}^3=\left\{\boldsymbol{u}\in\mathcal {C}_i^{'}:(u_{i_1},u_{i_2})=(1,t), t\in\{\pm 1,0\}\ \text{or}\ (u_{i_1},u_{i_2})=(0,1)\right\}.$$
\end{enumerate}

\begin{corollary}\label{field}
Let $\mathbb{P}=(M,R)$ be a pomset on $M=\left\{\lfloor\frac{m}{2}\rfloor/1,\lfloor\frac{m}{2}\rfloor/2,\ldots, \lfloor\frac{m}{2}\rfloor/2s\right\}$ with chain relation and $\pi$ be a labeling of the pomset $\mathbb{P}$ with $\pi(i)=2$ for $i\in[s]$. Let $\mathcal {C}$ be a linear $(\mathbb{P},\pi)$-code. If $\mathbb{Z}_m$ is a field (that is, $m$ is a prime number), then the $(\widetilde{\mathbb{P}},\pi)$-weight enumerator for $\mathcal {C}^{\bot}$ is given by:
$$W_{(\mathcal {C}^{\bot},\pi)}(x,y;\widetilde{\mathbb{P}})=x^{s\lfloor\frac{m}{2}\rfloor}+ \sum\limits_{i=1}^s\frac{2m^{2(s-i)}}{|\mathcal {C}|} \sum\limits_{j=1}^{\lfloor\frac{m}{2}\rfloor} \left(\frac{y}{x}\right)^{(s-i)\lfloor\frac{m}{2}\rfloor+j}$$
$$\Bigg[4jW_{(\mathcal {C}_i^{\bot},\pi)}(x,y;\mathbb{P})+(2m-4j)W_{({\mathcal {C}_{ij}^{1}}^{\bot},\pi)}(x,y;\mathbb{P})-4jW_{({\mathcal {C}_{ij}^{2}}^{\bot},\pi)}(x,y;\mathbb{P})+(m-4j)W_{({\mathcal {C}_{ij}^{3}}^{\bot},\pi)}(x,y;\mathbb{P})\Bigg].$$
\end{corollary}

\begin{proof}
By Theorem \ref{MAC}, we have
\begin{equation*}\label{equ1}
  \begin{aligned}
  & W_{(\mathcal{C}^{\bot},\pi)}(x,y;\widetilde{\mathbb{P}})=x^{s\lfloor\frac{m}{2}\rfloor}+ \sum\limits_{i=1}^s \frac{2m^{2(s-i)}}{|\mathcal {C}|}\sum \limits_{j=1}^{\lfloor\frac{m}{2}\rfloor}\left(\frac{y}{x}\right)^{(s-i) \lfloor\frac{m}{2}\rfloor+j}&\\
  & \left[4jW_{(\mathcal {C}_i,\pi)}(x,y;\mathbb{P})+\sum\limits_{u\in\mathcal {C}_i^{'}}\sum\limits_{a=1}^{2}\cos\frac{2\pi u_{i_a}j}{m}\prod\limits_{b<a}\left(1+2\sum\limits_{l=1}^{j-1}\cos\frac{2\pi u_{i_b}l}{m}\right)\prod\limits_{b>a}\left(1+2\sum\limits_{l=1}^{j}\cos\frac{2\pi u_{i_b}l}{m}\right)\right].&
  \end{aligned}
  \end{equation*}
Denote by
$$Z_{ij}^{\boldsymbol{u}}=\sum\limits_{a=1}^{2}\cos\frac{2\pi u_{i_a}j}{m}\prod\limits_{b<a}\left(1+2\sum\limits_{l=1}^{j-1}\cos\frac{2\pi u_{i_b}l}{m}\right)\prod\limits_{b>a}\left(1+2\sum\limits_{l=1}^{j}\cos\frac{2\pi u_{i_b}l}{m}\right).$$
For every $\boldsymbol{u}\in\mathcal {C}_i^{'}$, we have that $u_{i_1}$ and $u_{i_2}$ are not all $0$.
Take $\boldsymbol{u}\in\mathcal {C}_i^{'}$ whose $i$-th block $(u_{i_1},u_{i_2})$ has the form $(u_{i_1},u_{i_2})=(1,0)$. Since $\mathbb{Z}_m$ is a field, every element of $\mathbb{Z}_m\setminus\{0\}$ appears same times at $u_{i_1}$. Thus we have
\begin{eqnarray*}
\sum\limits_{\boldsymbol{u}\in\mathcal {C}_i^{'},u_{i_2}=0}Z_{ij}^{\boldsymbol{u}}&=&\sum\limits_{\boldsymbol{u}\in\mathcal {C}_i^{'},\atop (u_{i_1},u_{i_2})=(1,0)}\left[(2j+1)\sum\limits_{t=1}^{m-1}\cos\frac{2\pi tj}{m}+\sum\limits_{t=1}^{m-1}(1+2\sum\limits_{l=1}^{j-1}\cos\frac{2\pi tl}{m})\right]\\
&=&\sum\limits_{\boldsymbol{u}\in\mathcal {C}_i^{'},\atop (u_{i_1},u_{i_2})=(1,0)}\left[-2j-1+m-1+2\sum\limits_{l=1}^{j-1}\sum\limits_{t=1}^{m-1} \cos\frac{2\pi tl}{m}\right]\\
&=&\sum\limits_{\boldsymbol{u}\in\mathcal {C}_i^{'},\atop (u_{i_1},u_{i_2})=(1,0)}(m-2j-2-2j+2)\\
&=&(m-4j)\cdot\left|\left\{\boldsymbol{u}\in\mathcal {C}_i^{'}:(u_{i_1},u_{i_2})=(1,0)\right\}\right|.
\end{eqnarray*}
One can prove that $\sum\limits_{\boldsymbol{u}\in\mathcal {C}_i^{'},u_{i_1}=0}Z_{ij}=(m-4j)\cdot \left|\left\{\boldsymbol{u}\in\mathcal {C}_i^{'}:(u_{i_1},u_{i_2})=(0,1)\right\}\right|$ in the same way.

Suppose $\boldsymbol{u}\in\mathcal {C}_i^{'}$ satisfies that neither $u_{i_1}\neq 0$ nor $u_{i_2}\neq 0$. With out loss of generality, we assume that $(u_{i_1},u_{i_2})=(1,t)$ where $t\in\mathbb{Z}_m\setminus\{0\}$. Then
\begin{align}\label{unit}
\sum\limits_{\boldsymbol{v}\in \boldsymbol{u}\mathbb{Z}_m\setminus\{\textbf{0}\};\atop (u_{i_1},u_{i_2})=(1,t)}Z_{ij}^{\boldsymbol{v}}=&\sum\limits_{r=1}^{m-1}\cos\frac{2\pi rj}{m}\left(1+2\sum\limits_{l=1}^j\cos\frac{2\pi rtl}{m}\right)+\sum\limits_{r=1}^{m-1}\cos\frac{2\pi rtj}{m}\left(1+2\sum\limits_{l=1}^{j-1}\cos\frac{2\pi rl}{m}\right)\nonumber\\
=&\sum\limits_{r=1}^{m-1}\left[\cos\frac{2\pi rj}{m}\frac{e^{-i\frac{2\pi rtj}{m}}-e^{i\frac{2\pi rt(j+1)}{m}}}{1-e^{i\frac{2\pi rt}{m}}}+\cos\frac{2\pi rtj}{m}\frac{e^{-i\frac{2\pi r(j-1)}{m}}-e^{i\frac{2\pi rj}{m}}}{1-e^{i\frac{2\pi r}{m}}}\right]\nonumber\\
=&\sum\limits_{r=1}^{m-1}\sigma_r\left(\frac{\omega^j+\omega^{-j}}{2} \frac{\omega^{-tj}-\omega^{t(j+1)}}{1-\omega^t}\right)+ \sum\limits_{r=1}^{m-1}\sigma_r\left(\frac{\omega^{tj}+\omega^{-tj}}{2} \frac{\omega^{-(j-1)}-\omega^j}{1-\omega}\right),
\end{align}
where $\omega=e^{i\frac{2\pi}{m}}$ and $\{\sigma_r: 1\leq r\leq m-1\}$ is the Galois group of $m$-th cyclotomic field $\mathbb{Q}(\omega)$. Note that $\sum\limits_{\boldsymbol{v}\in \boldsymbol{u}\mathbb{Z}_m\setminus\{\textbf{0}\};\atop (u_{i_1},u_{i_2})=(1,t)}Z_{ij}^{\boldsymbol{v}}=\sum\limits_{\boldsymbol{v}\in \boldsymbol{u}\mathbb{Z}_m\setminus\{\textbf{0}\};\atop (u_{i_1},u_{i_2})=(1,-t)}Z_{ij}^{\boldsymbol{v}}$. The first summation in (\ref{unit}) becomes
\begin{eqnarray*}
&&\frac{1}{2}\sum\limits_{r=1}^{m-1}\sigma_r\left(\left(\omega^j+\omega^{-j})\omega^{-tj} (1+\omega^t+\omega^{2t}+\cdots+\omega^{2jt}\right)\right)\\
&=&\frac{1}{2}\sum\limits_{r=1}^{m-1}\sigma_r\left(\sum\limits_{k=0}^{2j}\omega^{j-tj+kt}+ \sum\limits_{k=0}^{2j}\omega^{-j-tj+kt}\right)\\
&=&\sum\limits_{k=0}^{2j}\sum\limits_{r=1}^{m-1}\sigma_r\left(\omega^{j-tj+kt}\right).
\end{eqnarray*}
The second summation in (\ref{unit}) becomes
\begin{eqnarray*}
&&\frac{1}{2}\sum\limits_{r=1}^{m-1}\sigma_r\left((\omega^{tj}+\omega^{-tj})\omega^{-(j-1)} (1+\omega+\omega^2+\cdots+\omega^{2j-2})\right)\\
&=&\frac{1}{2}\sum\limits_{r=1}^{m-1}\sigma_r\left(\sum\limits_{k=1}^{2j-1}\omega^{tj-j+k} +\sum\limits_{k=1}^{2j-1}\omega^{-tj-j+k}\right)\\
&=&\sum\limits_{k=1}^{2j-1}\sum\limits_{r=1}^{m-1}\sigma_r(\omega^{tj-j+k}).
\end{eqnarray*}
Hence
$$\sum\limits_{\boldsymbol{v}\in\boldsymbol{u}\mathbb{Z}_m\setminus\{\textbf{0}\};\atop (u_{i_1},u_{i_2})=(1,t)}Z_{ij}^{\boldsymbol{v}}= \sum\limits_{k=0}^{2j}\sum\limits_{r=1}^{m-1}\sigma_r(\omega^{j-tj+kt})+ \sum\limits_{k=1}^{2j-1}\sum\limits_{r=1}^{m-1}\sigma_r(\omega^{tj-j+k}).$$
It follows from this observation that
$$\sum\limits_{\boldsymbol{v}\in \boldsymbol{u}\mathbb{Z}_m\setminus\{\boldsymbol{0}\};\atop (u_{i_1},u_{i_2})=(1,\pm t)}Z_{ij}^{\boldsymbol{v}}=\left\{
                             \begin{array}{ll}
                              2m-4j, & \text{if}\ t\in\mathcal {C}_{ij}^1;\\[2mm]
                               -4j,   & \text{if}\ t\in\mathcal {C}_{ij}^2;\\[2mm]
                              m-4j, &\text{if}\ t\in\mathcal {C}_i^{'}\setminus
                              \mathcal {C}_{ij}^1\cup\mathcal {C}_{ij}^2.
                             \end{array}
                           \right.$$
The result then follows.
\end{proof}

\begin{example}
Let $M=\{2/1,2/2\}$ and $\mathbb{P}=(M,R)$ be a pomset whose order relation is chain relation. Let $\pi$ be a labeling of the pomset $\mathbb{P}$ such that $\pi(1)=\pi(2)=2$. Consider the $(\mathbb{P},\pi)$-code $\mathcal {C}\subseteq\mathbb{Z}_5^4$ generated by the following matrix:
$$\left(
  \begin{array}{cccc}
   1&0&1&1\\
   0&1&2&0
  \end{array}
\right).$$
Then it follows from Corollary \ref{field} that
$$W_{(\mathcal {C}^{\bot},\pi)}(x,y;\widetilde{\mathbb{P}})=x^4+\sum\limits_{i=1}^2\frac{2\times 5^{4-2i}}{25}x^{2i}y^{4-2i}\sum\limits_{j=1}^2\left(\frac{y}{x}\right)^j\left[4j|\mathcal {C}_i|+(10-4j)|\mathcal {C}_{ij}^1|-4j|\mathcal {C}_{ij}^2|+(5-4j)|\mathcal {C}_{ij}^3|\right].$$
Note that
$\mathcal {C}_{1}=\mathcal {C}_{2}=\{000\}$;
$\mathcal {C}_{11}^1=\mathcal {C}_{11}^2=\mathcal {C}_{11}^3=\mathcal {C}_{12}^1=\mathcal {C}_{12}^2=\mathcal {C}_{12}^3=\mathcal {C}_{21}^1=\mathcal {C}_{22}^2=\emptyset$;
$\mathcal {C}_{21}^2=\{2212,3413\}$;
$\mathcal {C}_{21}^3=\{1011,4114,0310,1201\}$;
$\mathcal {C}_{22}^1=\{2212,3413\}$;
$\mathcal {C}_{22}^3=\{1011,0310,1201,4114\}$.
Hence
\begin{eqnarray*}
W_{(\mathcal {C}^{\bot},\pi)}(x,y;\widetilde{\mathbb{P}})&=& x^4+2x^2y^2\left[4\left(\frac{y}{x}\right)+8\left(\frac{y}{x}\right)^2\right]+\\ &&\frac{2}{25}x^4\left[\left(\frac{y}{x}\right)(4+6|\mathcal {C}_{21}^1|-4|\mathcal {C}_{21}^2|+|\mathcal {C}_{21}^3|)+\left(\frac{y}{x}\right)^2(8+2|\mathcal {C}_{22}^1|-8|\mathcal {C}_{22}^2|-3|\mathcal {C}_{22}^3|\right]\\
&=&x^4+8xy^3+16y^4+\frac{2x^4}{25} \left[\left(\frac{y}{x}\right)(4-4\times2+4)+\left(\frac{y}{x}\right)^2(8+2\times2-3\times 4)\right]\\
&=&x^4+8xy^3+16y^4.
\end{eqnarray*}
On the other hand, the dual of $\mathcal {C}$ is
\begin{align*}
\mathcal {C}^{\bot}=&\{1004,2003,3002,4001,0123,0241,0314,0432,1122,2244,
3311,4433,\\
&1240,2430,3120,4310, 1313,2121,3434,4242,1431,2312,3243,4124\}
\end{align*}
and the $(\widetilde{\mathbb{P}},\pi)$-weight enumerator of $\mathcal {C}^{\bot}$ is
$$W_{(\mathcal {C}^{\bot},\pi)}(x,y;\widetilde{\mathbb{P}})=x^4+8xy^3+16y^4,$$
which coincides with Corollary \ref{field}.
$\hfill\square$
\end{example}

\section{MacWilliams type identities on direct and ordinal sum of general pomsets}

\quad\;For $i\in\{1,2\}$, let $(\mathbb{P}_i,\pi_i)$ be a pomset block structure over $[s_i]$ where $\pi_i:[s_i]\rightarrow\mathbb{N}$ is a map such that $\sum\limits_{j=1}^{s_i}\pi_i(j)=n_i$ and $\mathbb{P}_i=(M_i,R_i)$ with $M_i=\{\lfloor\frac{m}{2}\rfloor/1,\ldots,\lfloor\frac{m}{2}\rfloor/n_i\}$. Denote by $n=n_1+n_2$. Suppose that $\mathcal {C}_i$ is a linear $(\mathbb{P}_i,\pi_i)$-code. The direct sum of $\mathcal {C}_1$ and $\mathcal {C}_2$ denoted by $\mathcal {C}=\mathcal {C}_1\oplus\mathcal {C}_2$ is given by
$$\mathcal {C}=\{(\boldsymbol{u,v}): \boldsymbol{u}\in\mathcal {C}_1,\boldsymbol{v}\in\mathcal {C}_2\}.$$
Note that $\mathcal {C}$ is also a submodule of $V$.

Let $M=\left[\lfloor\frac{m}{2}\rfloor/1,\cdots,\lfloor\frac{m}{2}\rfloor/n_1, \lfloor\frac{m}{2}\rfloor/(n_1+1),\ldots,\lfloor\frac{m}{2}\rfloor/(n_1+n_2)\right]$. Define a pomset relation $R$ on $M$ in the following way:
$$p/i\ R\ q/j\Leftrightarrow \left(i,j\leq n_1\ \text{and}\ p/i\ R_1\ q/j\right)\ \text{or}\ \left(i,j>n_1\ \text{and}\ p/(i-n_1)\ R_2\ q/(j-n_1)\right).$$
for any $p/i$, $q/j\in M$. It is clear that $\mathbb{P}=(M,R)$ is a pomset and is called as \textbf{direct sum} of $\mathbb{P}_1$ and $\mathbb{P}_2$ denoted by $\mathbb{P}=\mathbb{P}_1\oplus\mathbb{P}_2$.

Define a pomset relation $R$ on $M$ as:
$$p/i\ R\ q/j\Leftrightarrow \left(i,j\leq n_1\ \text{and}\ p/i\ R_1\ q/j\right)\ \text{or}\ \left(i,j>n_1\ \text{and}\ p/(i-n_1)\ R_2\ q/(j-n_1)\right)\ \text{or}\ \left(i\leq n_1<j\right)$$
for any $p/i$, $q/j\in M$. Then $\mathbb{P}=(M,R)$ is a pomset and is called as \textbf{ordinal sum} of $\mathbb{P}_1$ and $\mathbb{P}_2$ denoted by $\mathbb{P}=\mathbb{P}_1+\mathbb{P}_2$. See [\ref{POMSET}] for detailed discussion on sum of pomsets.

With these definitions, we have that $\widetilde{\mathbb{P}_1\oplus\mathbb{P}_2}=\widetilde{\mathbb{P}}_1\oplus \widetilde{\mathbb{P}}_2$ and $\widetilde{\mathbb{P}_1+\mathbb{P}_2}=\widetilde{\mathbb{P}}_2+\widetilde{\mathbb{P}}_1$. Define the \textbf{sum} of $\pi_1$ and $\pi_2$ denoted by $\pi=\pi_1\oplus\pi_2$ as $\pi:[s_1+s_2]\rightarrow\mathbb{N}$ such that
$$\pi(i)=\left\{
                             \begin{array}{ll}
                             \pi_1(i) &\text{if}\ i\leq n_1;\\[2mm]
                             \pi_2(i-n_1)&\text{if}\ i>n_1.
                             \end{array}
                           \right.$$

We now consider the code $\mathcal {C}$ equipped with $(\mathbb{P}_1\oplus\mathbb{P}_2,\pi)$ and $(\mathbb{P}_1+\mathbb{P}_2,\pi)$ block structures respectively.

With notations introduced above, we obtain the following result.

\begin{theorem}
(1) For a linear $(\mathbb{P}_1\oplus\mathbb{P}_2,\pi)$-code $\mathcal {C}=\mathcal {C}_1\oplus\mathcal {C}_2$, we have
$$W_{(\mathcal {C}^{\bot},\pi)}(x,y;\widetilde{\mathbb{P}_1\oplus\mathbb{P}_2})=W_{(\mathcal {C}_1^{\bot},\pi_1)}(x,y;\widetilde{\mathbb{P}}_1)W_{(\mathcal {C}_2^{\bot},\pi_2)}(x,y;\widetilde{\mathbb{P}}_2).$$
(2) For a linear $(\mathbb{P}_1+\mathbb{P}_2,\pi)$-code $\mathcal {C}=\mathcal {C}_1\oplus\mathcal {C}_2$, we have
$$W_{(\mathcal {C}^{\bot},\pi)}(x,y;\widetilde{\mathbb{P}_1+\mathbb{P}_2})=x^{s_1\lfloor\frac{m}{2}\rfloor} W_{(\mathcal {C}_2^{\bot},\pi_2)}(x,y;\widetilde{\mathbb{P}}_2)+\frac{m^{n_2}}{|\mathcal {C}_2|}y^{s_2\lfloor\frac{m}{2}\rfloor}\left(W_{(\mathcal {C}_1^{\bot},\pi_1)}(x,y;\widetilde{\mathbb{P}}_1)-x^{s_1\lfloor\frac{m}{2}\rfloor}\right).$$
\end{theorem}

\begin{proof}
Set $\mathbb{P}=\mathbb{P}_1\oplus\mathbb{P}_2$. Define a function $f:\mathbb{Z}_{m}^n\rightarrow\mathbb{C}[x,y]$ by
$$f(\boldsymbol{u})= x^{(s_1+s_2)\lfloor\frac{m}{2}\rfloor-w_{(\widetilde{\mathbb{P}},\pi)}(\boldsymbol{u})} y^{w_{(\widetilde{\mathbb{P}},\pi)}(\boldsymbol{u})}.$$
Then the Fourier transform $\hat{f}$ of $f$ is
$$\hat{f}(\boldsymbol{u})= \sum\limits_{(\boldsymbol{v_1,v_2})\in\mathbb{Z}_m^n}\chi(\boldsymbol{u}\cdot\boldsymbol{v}) x^{(s_1+s_2)\lfloor\frac{m}{2}\rfloor-w_{(\widetilde{\mathbb{P}},\pi)}(\boldsymbol{v})} y^{w_{(\widetilde{\mathbb{P}},\pi)}(\boldsymbol{v})}.$$
The $(\widetilde{\mathbb{P}},\pi)$-weight enumerator of $\mathcal {C}^{\bot}$ is
\begin{eqnarray*}
&&W_{(\mathcal {C}^{\bot},\pi)}(x,y;\widetilde{\mathbb{P}})=\sum\limits_{\boldsymbol{u}\in\mathcal {C}^{\bot}}f(\boldsymbol{u})=\frac{1}{|\mathcal {C}|}\sum\limits_{\boldsymbol{u}\in\mathcal {C}}\hat{f}(\boldsymbol{u})\\
&=&\frac{1}{|\mathcal {C}|}\sum\limits_{\boldsymbol{u}\in\mathcal {C}}\sum\limits_{(\boldsymbol{v_1,v_2})\in\mathbb{Z}_m^n} \chi(\boldsymbol{u}\cdot\boldsymbol{v}) x^{(s_1+s_2)\lfloor\frac{m}{2}\rfloor-w_{(\widetilde{\mathbb{P}},\pi)}(\boldsymbol{v})} y^{w_{(\widetilde{\mathbb{P}},\pi)}(\boldsymbol{v})}\\
&=&\frac{1}{|\mathcal {C}|}\sum\limits_{\boldsymbol{u}\in\mathcal {C}}\sum\limits_{\boldsymbol{v}\in\mathbb{Z}_m^n}\chi(\boldsymbol{u_1}\cdot\boldsymbol{v_1}) \chi(\boldsymbol{u_2}\cdot\boldsymbol{v_2}) x^{(s_1+s_2)\lfloor\frac{m}{2}\rfloor-w_{(\widetilde{\mathbb{P}}_1,\pi_1)}(\boldsymbol{v_1})- w_{(\widetilde{\mathbb{P}}_2,\pi_2)}(\boldsymbol{v_2})} y^{w_{(\widetilde{\mathbb{P}}_1,\pi_1)}(\boldsymbol{v_1})+ w_{(\widetilde{\mathbb{P}}_2,\pi_2)}(\boldsymbol{v_2})}\\
&=&\prod\limits_{i=1}^2\frac{1}{|\mathcal {C}_i|}\sum\limits_{\boldsymbol{u_i}\in\mathcal {C}_i}\sum\limits_{\boldsymbol{v_i}\in\mathbb{Z}_m^{n_i}} \chi(\boldsymbol{u_i}\cdot\boldsymbol{v_i}) x^{s_i\lfloor\frac{m}{2}\rfloor-w_{(\widetilde{\mathbb{P}}_i,\pi_i)}(\boldsymbol{v_i})} y^{w_{(\widetilde{\mathbb{P}}_i,\pi_i)}(\boldsymbol{v_i})}\\
&=&W_{(\mathcal {C}_1^{\bot},\pi_1)}(x,y;\widetilde{\mathbb{P}}_1)W_{(\mathcal {C}_2^{\bot},\pi_2)}(x,y;\widetilde{\mathbb{P}}_2).
\end{eqnarray*}
If $\mathbb{P}=\mathbb{P}_1+\mathbb{P}_2$, then
\begin{eqnarray}\label{ordinal}
\sum\limits_{\boldsymbol{u}\in\mathcal {C}}\hat{f}(\boldsymbol{u})&=&\sum\limits_{\boldsymbol{u}\in\mathcal {C}}\sum\limits_{(\boldsymbol{v_1,v_2})\in\mathbb{Z}_m^n}\chi(\boldsymbol{u}\cdot\boldsymbol{v}) x^{(s_1+s_2)\lfloor\frac{m}{2}\rfloor-w_{(\widetilde{\mathbb{P}},\pi)}(\boldsymbol{v})} y^{w_{(\widetilde{\mathbb{P}},\pi)}(\boldsymbol{v})}\nonumber\\
&=&\sum\limits_{\boldsymbol{u}\in\mathcal {C}}\left[\sum\limits_{(\textbf{0},\boldsymbol{v_2})\in\mathbb{Z}_m^n} \chi(\boldsymbol{u_2}\cdot\boldsymbol{v_2}) x^{(s_1+s_2)\lfloor\frac{m}{2}\rfloor-w_{(\widetilde{\mathbb{P}}_2,\pi_2)}(\boldsymbol{v_2})} y^{w_{(\widetilde{\mathbb{P}}_2,\pi_2)}(\boldsymbol{v_2})}+\right.\nonumber\\
&&\left.\sum\limits_{(\boldsymbol{v_1,v_2})\in\mathbb{Z}_m^n\atop \boldsymbol{v_1}\neq \boldsymbol{0}}\chi(\boldsymbol{u_1}\cdot\boldsymbol{v_1}) \chi(\boldsymbol{u_2}\cdot\boldsymbol{v_2}) x^{(s_1+s_2)\lfloor\frac{m}{2}\rfloor- w_{(\widetilde{\mathbb{P}}_1,\pi_1)}(\boldsymbol{v_1})-s_2\lfloor\frac{m}{2}\rfloor} y^{w_{(\widetilde{\mathbb{P}}_1,\pi_1)}(\boldsymbol{v_1})+s_2\lfloor\frac{m}{2}\rfloor}\right] \nonumber\\
&=&\sum\limits_{\boldsymbol{u}\in\mathcal {C}}x^{s_1\lfloor\frac{m}{2}\rfloor}\sum\limits_{\boldsymbol{v_2}\in\mathbb{Z}_m^{n_2}} \chi(\boldsymbol{u_2}\cdot\boldsymbol{v_2})
x^{s_2\lfloor\frac{m}{2}\rfloor-w_{(\widetilde{\mathbb{P}}_2,\pi_2)}(\boldsymbol{v_2})} y^{w_{(\widetilde{\mathbb{P}}_2,\pi_2)}(\boldsymbol{v_2})}+\nonumber\\
&&\sum\limits_{\boldsymbol{u}\in\mathcal {C}}y^{s_2\lfloor\frac{m}{2}\rfloor}\sum\limits_{(\boldsymbol{v_1,v_2})\in\mathbb{Z}_m^n\atop \boldsymbol{v_1}\neq\boldsymbol{0}}\chi(\boldsymbol{u_1}\cdot\boldsymbol{v_1}) \chi(\boldsymbol{u_2}\cdot\boldsymbol{v_2}) x^{s_1\lfloor\frac{m}{2}\rfloor-w_{(\widetilde{\mathbb{P}}_1,\pi_1)}(\boldsymbol{v_1})} y^{w_{(\widetilde{\mathbb{P}}_1,\pi_1)}(\boldsymbol{v_1})}.
\end{eqnarray}
The first summation in (\ref{ordinal}) is
\begin{eqnarray*}
&&x^{s_1\lfloor\frac{m}{2}\rfloor} \sum\limits_{(\boldsymbol{u_1,u_2})\in\mathcal {C}}\sum\limits_{\boldsymbol{v_2}\in\mathbb{Z}_m^{n_2}} \chi(\boldsymbol{u_2}\cdot\boldsymbol{v_2})
x^{s_2\lfloor\frac{m}{2}\rfloor-w_{(\widetilde{\mathbb{P}}_2,\pi_2)}(\boldsymbol{v_2})} y^{w_{(\widetilde{\mathbb{P}}_2,\pi_2)}(\boldsymbol{v_2})}\\
&=&|\mathcal {C}_1|x^{s_1\lfloor\frac{m}{2}\rfloor}\sum\limits_{\boldsymbol{u_2}\in\mathcal {C}_2}\sum\limits_{\boldsymbol{v_2}\in\mathbb{Z}_m^{n_2}} \chi(\boldsymbol{u_2}\cdot\boldsymbol{v_2})
x^{s_2\lfloor\frac{m}{2}\rfloor-w_{(\widetilde{\mathbb{P}}_2,\pi_2)}(\boldsymbol{v_2})} y^{w_{(\widetilde{\mathbb{P}}_2,\pi_2)}(\boldsymbol{v_2})}\\
&=&|\mathcal {C}|x^{s_1\lfloor\frac{m}{2}\rfloor}\sum\limits_{\boldsymbol{v_2}\in\mathcal {C}_2^{\bot}} x^{s_2\lfloor\frac{m}{2}\rfloor-w_{(\widetilde{\mathbb{P}}_2,\pi_2)}(\boldsymbol{v_2})} y^{w_{(\widetilde{\mathbb{P}}_2,\pi_2)}(\boldsymbol{v_2})}\\
&=&|\mathcal {C}|x^{s_1\lfloor\frac{m}{2}\rfloor}W_{(\mathcal {C}_2^{\bot},\pi_2)}(x,y;\widetilde{\mathbb{P}}_2).
\end{eqnarray*}
The second summation in (\ref{ordinal}) is
\begin{eqnarray*}
&&y^{s_2\lfloor\frac{m}{2}\rfloor}\sum\limits_{\boldsymbol{u}\in\mathcal {C}}\sum\limits_{\textbf{0}\neq \boldsymbol{v_1}\in\mathbb{Z}_m^{n_1}}\chi(\boldsymbol{u_1}\cdot\boldsymbol{v_1}) x^{s_1\lfloor\frac{m}{2}\rfloor-w_{(\widetilde{\mathbb{P}}_1,\pi_1)}(\boldsymbol{v_1})} y^{w_{(\widetilde{\mathbb{P}}_1,\pi_1)}(\boldsymbol{v_1})} \sum\limits_{\boldsymbol{v_2}\in\mathbb{Z}_m^{n_2}} \chi(\boldsymbol{u_2}\cdot\boldsymbol{v_2})\\
&=&m^{n_2} y^{s_2\lfloor\frac{m}{2}\rfloor}\sum\limits_{(\boldsymbol{u_1},\textbf{0})\in\mathcal {C}}\sum\limits_{\textbf{0}\neq \boldsymbol{v_1}\in\mathbb{Z}_m^{n_1}}\chi(\boldsymbol{u_1}\cdot\boldsymbol{v_1}) x^{s_1\lfloor\frac{m}{2}\rfloor-w_{(\widetilde{\mathbb{P}}_1,\pi_1)}(\boldsymbol{v_1})} y^{w_{(\widetilde{\mathbb{P}}_1,\pi_1)}(\boldsymbol{v_1})}\\
&=&m^{n_2}y^{s_2\lfloor\frac{m}{2}\rfloor}\sum\limits_{\boldsymbol{u_1}\in\mathcal {C}_1}\left(\sum\limits_{\boldsymbol{v_1}\in\mathbb{Z}_m^{n_1}} \chi(\boldsymbol{u_1}\cdot\boldsymbol{v_1}) x^{s_1\lfloor\frac{m}{2}\rfloor-w_{(\widetilde{\mathbb{P}}_1,\pi_1)}(\boldsymbol{v_1})} y^{w_{(\widetilde{\mathbb{P}}_1,\pi_1)}(\boldsymbol{v_1})}- x^{s_1\lfloor\frac{m}{2}\rfloor}\right)\\
&=&m^{n_2}y^{s_2\lfloor\frac{m}{2}\rfloor}\left(|\mathcal {C}_1|\sum\limits_{\boldsymbol{v_1}\in\mathcal {C}_1^{\bot}}x^{s_1\lfloor\frac{m}{2}\rfloor- w_{(\widetilde{\mathbb{P}}_1,\pi_1)}(\boldsymbol{v_1})} y^{w_{(\widetilde{\mathbb{P}}_1,\pi_1)}(\boldsymbol{v_1})}-|\mathcal {C}_1|x^{s_1\lfloor\frac{m}{2}\rfloor}\right)\\
&=&|\mathcal {C}_1|m^{n_2}y^{s_2\lfloor\frac{m}{2}\rfloor}\left(W_{(\mathcal {C}_1^{\bot},\pi_1)}(x,y;\widetilde{\mathbb{P}}_1)-x^{s_1\lfloor\frac{m}{2}\rfloor}\right).
\end{eqnarray*}
Hence
$$W_{(\mathcal {C}^{\bot},\pi)}(x,y;\widetilde{\mathbb{P}})=\frac{1}{|\mathcal {C}|}\sum\limits_{u\in\mathcal {C}}\hat{f}(u)=x^{s_1\lfloor\frac{m}{2}\rfloor} W_{(\mathcal {C}_2^{\bot},\pi_2)}(x,y;\widetilde{\mathbb{P}}_2)+\frac{m^{n_2}}{|\mathcal {C}_2|}y^{s_2\lfloor\frac{m}{2}\rfloor}\left(W_{(\mathcal {C}_1^{\bot},\pi_1)}(x,y;\widetilde{\mathbb{P}}_1)-x^{s_1\lfloor\frac{m}{2}\rfloor}\right).$$
\end{proof}

Let $\mathcal {C}_i\subseteq\mathbb{Z}_m^{n_i}$ be a linear $(\mathbb{P}_i,\pi_i)$-code for $i\in\{1,2,\ldots,\lambda\}$. Let
$\mathcal {C}=\bigoplus\limits_{i=1}^{\lambda}\mathcal {C}_i$ and $\pi=\bigoplus\limits_{i=1}^{\lambda}\pi_i$. By the induction on $\lambda$, we obtain the following result.
\begin{theorem}\label{general}
\begin{enumerate}[(1)]
\item If $\mathbb{P}=\bigoplus\limits_{i=1}^{\lambda}\mathbb{P}_i$, then $\mathcal {C}$ is a linear $(\mathbb{P},\pi)$-code and
    $$W_{(\mathcal {C}^{\bot},\pi)}(x,y;\widetilde{\mathbb{P}})= \prod\limits_{i=1}^{\lambda}
    W_{(\mathcal{C}_i^{\bot},\pi_i)}(x,y;\widetilde{\mathbb{P}}_i).$$
\item Set $s=s_1+\cdots+s_{\lambda}$. If $\mathbb{P}=\mathbb{P}_1+\mathbb{P}_2+\cdots+\mathbb{P}_{\lambda}$, then $\mathcal {C}$ is a linear $(\mathbb{P},\pi)$-code and
    \begin{eqnarray*}
    W_{(\mathcal {C}^{\bot},\pi)}(x,y;\widetilde{\mathbb{P}})&=& x^{s_1+\cdots+s_{\lambda-1}\lfloor\frac{m}{2}\rfloor}W_{(\mathcal {C}_{\lambda}^{\bot},\pi_{\lambda})}(x,y;\widetilde{\mathbb{P}}_{\lambda})+\\[2mm] &&\prod\limits_{i=1}^{\lambda} \frac{m^{n_{i+1}+\cdots+n_{\lambda}}y^{s_{i+1}+\cdots+s_{\lambda}}}{|\mathcal {C}_{i+1}|\cdots|\mathcal {C}_{\lambda}|x^{s_i+\cdots+s_{\lambda}-s}}\left(W_{(\mathcal {C}_i^{\bot},\pi_i)}(x,y;\widetilde{\mathbb{P}}_i)-x^{s_i\lfloor\frac{m}{2}\rfloor}\right).
    \end{eqnarray*}
\end{enumerate}
\end{theorem}

\begin{remark}
When the structure of blocks are trivial (that is, $\pi(i)=1$ for $i\in[s]$), then we have the result for pomset metric stated in [\ref{POMSETMAC}].
\end{remark}

\section{Conclusion}

\quad\;In this paper, we consider the relation between the weight enumerators of a code and its dual when the pomset is a chain pomset. Note that when the metric is taken to be poset metric over $\mathbb{F}_q^n$ where $\mathbb{F}_q$ is a finite field, it is known that a poset $P$ admits the MacWilliams identity if and only if $P$ is a hierarchical poset (see [\ref{POSETMAC}]) . Further, when the metric is considered to be a poset block metric over $\mathbb{F}_q^n$, it is known that a poset-block space admits a MacWilliams type identity if and only if the poset is hierarchical, and at any level of the poset, all the blocks have the same dimension (see [\ref{posetmac}]). Nevertheless, being a chain pomset can not be a necessary and sufficient condition for a pomset to admit the MacWilliams identity.

Based on these observations, a natural question is brought up: is hierarchical a necessary condition for a pomset $\mathbb{P}=(M,R)$ to admit the MacWilliams identity (here hierarchical pomset means that $M$ can be partitioned into $l$ nonempty anti-chains in $\mathbb{P}$, say, $(A_1,\ldots,A_l)$, such that for any $i,j\in[s]$ with $i<j$, it holds that $p/a\ R\ q/b$ for all $a\in A_i$ and $b\in A_j$) ?

\end{document}